\documentclass[conference]{IEEEtran}

\usepackage{multirow}
\usepackage{epsfig}
\usepackage{subfigure}
\usepackage{cite}
\usepackage{amsmath}
\usepackage{amssymb}
\usepackage{footnote}
\usepackage{array}

\newtheorem{theorem}{\textbf{Theorem}}

\newtheorem{corollary}{\textbf{Corollary}}

\begin{document}

\title{Capacity Scaling Law by Multiuser Diversity in Cognitive Radio Systems}

\author{\IEEEauthorblockN{Jun-pyo Hong and Wan Choi}
\IEEEauthorblockA{Dept. of Electrical Engineering, Korea Advanced Institute of Science and Technology (KAIST)\\
Email: \{jp\_hong, w.choi\}@kaist.ac.kr}} \maketitle


\begin{abstract}
This paper analyzes the multiuser diversity gain in a cognitive radio (CR) system where secondary transmitters opportunistically utilize the spectrum licensed to primary users only when it is not occupied by the primary users. To protect the primary users from the interference caused by the missed detection of primary transmissions in the secondary network, minimum average throughput of the primary network is guaranteed by transmit power control at the secondary transmitters. The traffic dynamics of a primary network are also considered in our analysis. We derive the average achievable capacity of the secondary network and analyze its asymptotic behaviors to characterize the multiuser diversity gains in the CR system.
\end{abstract}

\section{Introduction}
The demand for wireless spectrum is constantly increasing as various
wireless communication services have appeared. Correspondingly the
available spectrum has become scarce. Current spectrum allocation
policies aggravate spectrum scarcity since a particular spectrum is
dedicated to only specific users as Federal Communications
Commission (FCC) reported \cite{FCC_Report}. As a solution to
improve the spectrum efficiency, \emph{cognitive radio} (CR) has
emerged, where secondary (unlicensed) users try to access the
spectrum without interfering with communications of primary users.
The term CR is classified into several techniques depending on the way to access the spectrum licensed to the primary users:
\emph{overlaid}, \emph{underlaid}, and \emph{interweaved} CR systems
\cite{Mag_Jafar}.
Our work focus on interweaved CR system where the secondary users are allowed to access the
spectrum licensed to primary users only when primary users do not
occupy the spectrum. This approach comes from the idea of
\emph{opportunistic communication} \cite{Mitola}. Secondary users
monitor whether the spectrum is occupied by primary users in order
to opportunistically communicate over vacant spectrum without
interfering with primary users.
Stable throughput of an
interweaved CR system was analyzed by reflecting random packet
arrival of a primary user using queueing process in \cite{Osvaldo}
and \cite{Osvaldo_1}.
Srinivasa and Jafar \cite{CR_Jafar} studied the optimum number of
secondary users that maximizes the total throughput in a
decentralized CR system.

Secondary user scheduling and medium access control play key roles
in CR systems but they have not been well studied yet. It has been
well known that opportunistic user selection for transmission
provides a multiuser diversity gain due to fluctuations of fading
channels \cite{Viswanth2002}. For non-CR systems, there have been
many studies on characterizing the multiuser diversity gains
\cite{MUD_1, MUD_2, MUD_3, Choi2008}. These studies showed that the
multiuser diversity gain in terms of average capacity grows like
$\log_2(\ln N)$ and $\sqrt{2\log_2 N} $ in Rayleigh fading channels
and lognormal shadowing channels, respectively, where $N$ is the
number of users. The multiuser diversity gain of a secondary network
in an underlaid CR system was recently discussed in \cite{MUD_SS,
MUD_SS_Liang}. It was shown that the average capacity of the
secondary network scales like $\log_2(\ln N)$ and $\log_2 N$ under
the finite and the infinite peak transmit power constraints at the
secondary transmitters, respectively. Moreover, the scaling laws of
underlaid CR system in a cognitive ad hoc network was studied in
\cite{Scalinglaw_Tarokh} although it did not focus on the multiuser
diversity gain coming from opportunistic user selection. Despite the
fact that the term CR typically refers to interweaved CR systems,
however, the multiuser diversity gain in interweaved CR systems has
not been identified.

In this paper, we investigate characteristics of multiuser diversity
gains in interweaved CR systems by deriving average achievable
capacity of a secondary network. Multiuser scheduling and transmit
power control are employed in a secondary network to maximize the
achievable capacity of a secondary user network and to satisfy a
quality of service (Qos) constraint of a primary user network.
Unlike underlaid CR systems, secondary users are not allowed to
access the spectrum licensed to primary users when a secondary
network detects primary users' transmission. So our analysis
reflects both traffic dynamics of a primary user network and sensing
reliability in a secondary network. Furthermore, this paper takes
into account the interference from a primary transmitter contrary to
\cite{MUD_SS}. Our asymptotic analysis and numerical results
insightfully capture the key characteristics of multiuser diversity
gains in interweaved CR systems. It is shown that multiuser
diversity gains in interweaved CR systems are quite different from
those in non-CR based systems due to the QoS constraint on a primary
user network. Moreover, even if the secondary users of interweaved
and underlaid CR systems have different ways to access the spectrum
licensed to primary users, the multiuser diversity gains of both CR
systems show very similar asymptotic characteristics.



\vspace{0.05in}
\section{System Description}

\subsection{System and Channel Model}

\begin{figure}[!t]
    \centerline{\psfig{figure=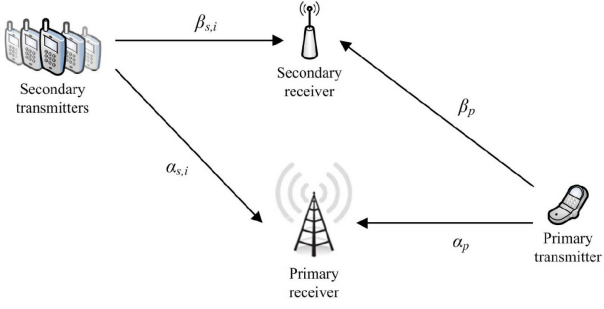,width=0.9\columnwidth} }
    \caption{System model}
    \label{Fig:System_Model}
\end{figure}

The system model considered in this paper is illustrated in Fig.
\ref{Fig:System_Model}, which consists of a primary receiver
(primary base station), a primary transmitter, a secondary receiver
(secondary base station), and $N$ secondary transmitters which try
to access the spectrum licensed to a primary transmitter. The
additive thermal noise is assumed to be zero mean mutually
independent, circularly symmetric, complex Gaussian random variable
with unit variance. The channel gains from the $i$th secondary
transmitter to the primary receiver and to the secondary receiver
are denoted by $\alpha_{s,i}$ and $\beta_{s,i}$, respectively.
Similarly, $\alpha_{p}$ and $\beta_{p}$ are the channel gains from
the primary transmitter to the primary receiver and to the secondary
receiver, respectively. All channel gains are assumed to be
independent and identically distributed (i.i.d.) exponential random
variable with unit mean in Rayleigh flat fading channels. It is also
assumed that the $i$th secondary transmitter knows the instantaneous
channel state information of $\alpha_{s,i}$ and $\beta_{s,i}$ and
knows the statistical channel information of $\alpha_{p}$ and
$\beta_{p}$. The secondary transmitters are able to obtain the
instantaneous channel state information of $\alpha_{s,i}$ and
$\beta_{s,i}$ through a periodic sensing of pilot signal from the
primary and secondary receivers, according to the channel
reciprocity.


In our system model, all the secondary transmitters which have
packets to transmit compute the maximum allowable transmit power
$P_{s,i}$ based on $\alpha_{s,i}$ to maximize their throughput while
ensuring the QoS of primary users. Then, the secondary transmitters
report their 1-bit sensing results with calculated transmit power
$P_{s,i}$. The secondary receiver fuses received sensing results
from the secondary transmitters and makes a global decision.
The reliability of cooperative sensing is captured by detection probability $p_d$ and false alarm probability $p_f$.
However, detailed discussion about how to fuse the received sensing
results and how to reflect the performance of cooperative sensing
are beyond the scope of this paper. If the global decision declares
the absence of primary users' traffic in the spectrum, the secondary
receiver selects the secondary transmitter which has the largest
signal-to-noise power ratio (SNR) $\beta_{s,i}P_{s,i}$ to maximize
secondary network throughput. Otherwise, all secondary transmitters
are silent during a unit time slot. On the other hand, a primary
transmitter transmits its packets with its maximum power, $P_{p}$.

\subsection{Traffic Model}
We consider a \emph{packet-by-packet} access network where the spectrum licensed to the primary users is accessed in a time slotted manner similarly in \cite{Osvaldo}.
In other words, all primary and secondary transmitters can transmit one packet per each time slot.
The traffic at the primary transmitter is modeled by a
Bernoulli process with average packet arrival rate,
$0\le\lambda\le1$ [packets/slot], and departure rate, $0\le\mu\le1$
[packets/slot]. Note that the departure rate can be interpreted as
the probability of successful packet transmission during a time
slot; moreover, it can directly be translated into average
throughput of a primary user. In contrast to arrival rate $\lambda$,
departure rate $\mu$ is significantly affected by the behavior of
the secondary network.
So the primary user sets a minimum tolerable
departure rate $\mu_{\min}$ to guarantee a certain service rate for
the case of missed detection in the secondary network by
\begin{eqnarray}\label{min_departure}
\mu \ge \mu_{\min}.
\end{eqnarray}
It should be noted that the departure rate without any interference
is greater than the arrival rate for an essentially stable primary
user's queue.

\vspace{0.05in}
\section{QoS Constraint and Secondary User Scheduling}
\subsection{Departure Rate at a Primary Transmitter}

The outage probabilities when there is no interference and when
there is interference from the secondary network are given,
respectively, by
\begin{align}
P_{out} &= \Pr\left[ \log_2\left(1+\alpha_pP_p\right)<R\right]\label{P_outage1} \\
P_{out}^{inf} &= \Pr\left[\log_2\left(1+\frac{\alpha_pP_p}{1+\alpha_{s,i}P_{s,i}}\right)<R\right],
\label{P_outage2}
\end{align}
where $R$ is the required rate at the primary receiver. Then, the
departure rate of the primary transmitter is given by
\begin{eqnarray}\label{departure}
\mu = p_d(1-P_{out})+(1-p_d)(1-P_{out}^{inf}).
\end{eqnarray}
From a perspective of a primary user, a minimum departure rate,
$\mu_{\min}$, needs to be set as a OoS parameter and is notified to
the secondary network so that interference from a secondary network
is regulated to guarantee the minimum departure rate.

\subsection{Power Adaptation and Secondary User Scheduling}
Based on the channel information including the instantaneous value
of $\alpha_{s,i}$ and the statistical properties of $\alpha_p$, each
secondary transmitter can estimate the corresponding outage
probabilities of the primary user, when it is scheduled.
The outage probabilities, (\ref{P_outage1}) and (\ref{P_outage2}),
can be rewritten by
\begin{align}
P_{out} &= \Pr\left[\alpha_p<\frac{2^R-1}{P_p}\right] = 1-e^{-R_p}
\label{out_noinf}
\end{align}
\begin{align}
P_{out}^{inf} &= \Pr\left[\alpha_p<\left(2^R-1\right)\frac{1+\alpha_{s,i}P_{s,i}}{P_p}\right]\nonumber\\
&= 1-e^{-R_p(1+\alpha_{s,i}P_{s,i})}
\label{out_inf}
\end{align}
where $R_p=(2^R-1)/P_p$.
If the $i$th secondary transmitter is scheduled for transmission,
the departure rate of the primary user is obtained by plugging
(\ref{out_noinf}) and (\ref{out_inf}) into (\ref{departure}),
\begin{eqnarray}\label{mu_func}
\mu(\alpha_{s,i}, P_{s,i}) = p_de^{-R_p} + (1-p_d)e^{-R_p(1+\alpha_{s,i}P_{s,i})}.
\end{eqnarray}
The departure rate of the primary user is affected by transmit power
of the secondary transmitter $P_{s,i}$ and instantaneous
interference channel gain $\alpha_{s,i}$.
Hereinafter, we use $\mu(\alpha_{s,i}, P_{s,i})$ instead of $\mu$ to
represent the departure rate as a function of $\alpha_{s,i}$ and
$P_{s,i}$.

The transmit power of the $i$th secondary transmitter is determined
by the value making $\mu(\alpha_{s,i}, P_{s,i})$ equal to
$\mu_{\min}$ since $\mu(\alpha_{s,i}, P_{s,i})$ is inversely
proportional to $P_{s,i}$. The power $P_{\mu,i}$ corresponding to
$\mu_{\min}$ is given by
\begin{eqnarray}\label{P_mu}
P_{\mu,i}=\frac{1}{\alpha_{s,i}}\left[\frac{1}{R_p}\ln\left(\frac{1-p_d}{\mu_{\min}-p_de^{-R_p}}\right)-1\right].
\end{eqnarray}

Since the transmit power of the secondary user is limited by a
maximum power constraint, i.e., $P_{s,i}\le P_{s,\max}$, the
transmit power of the $i$th secondary transmitter is given by
\begin{align} \label{P_si}
P_{s,i}=\left\{ \begin{array}{ll}
0 ~,& P_{\mu,i}<0 \\
P_{\mu,i} ~,& 0\le P_{\mu,i}<P_{s,\max}\\
P_{s,\max} ~,& P_{\mu,i} \geq P_{s,\max}
\end{array} \right.
.
\end{align}
The first case $P_{\mu,i}<0$ corresponds to the situation where the
primary user cannot achieve $\mu_{\min}$ even if a secondary network
do not interfere with the primary user. In this case, all the
secondary transmitters can not transmit anything, regardless of the
sensing result, so that there is no secondary network throughput. In
the other cases, a secondary transmitter can transmit a packet with
adaptively adjusted power based on the value of $\alpha_{s,i}$. Our
analysis focuses on the last two cases where the spectrum licensed
to primary user is opportunistically accessed.
Hence, from (\ref{P_mu}), it is assumed that
\begin{align}
K = \frac{1}{R_p}\ln\left(\frac{1-p_d}{\mu_{\min}-p_de^{-R_p}}\right)-1>0
\end{align}
where a positive real constant $K$ is pre-determined by the system
parameters.

Based on the computed $P_{s,i}$ in (\ref{P_si}), all the secondary
transmitters report their 1-bit sensing results with its transmit
power, $P_{s,i}$. Using the collected sensing results, the secondary
receiver checks the availability of the spectrum licensed to the
primary network. If the spectrum is decided to be vacant, the
secondary receiver selects the $i^*$th secondary transmitter whose
SNR is the largest among $N$ secondary transmitters such as
\begin{align}
i^* = \arg\max_i \gamma_{s,i}~
\end{align}where the received SNR $\gamma_{s,i}$ from the
$i$th secondary transmitter is given by
\begin{align}\label{class_SU_SNR}
\gamma_{s,i}  = \left\{ \begin{array} {ll}
P_{\mu,i}\beta_{s,i}=\frac{K\beta_{s,i}}{\alpha_{s,i}} ~,& P_{\mu,i}<P_{s,\max}\\
P_{s,\max}\beta_{s,i} ~,& P_{\mu,i}\ge P_{s,\max}
\end{array} \right.
.
\end{align}
\vspace{0.05in}
\section{Characteristics of Multiuser Diversity in a Cognitive Radio System}
Unfortunately, analyzing the average capacity of secondary network is not easy to handle because extracting the PDFs of $\alpha_{s,i^*}$ and $\beta_{s,i^*}$from that of $\gamma_{s,i^*}$ is intractable. Instead, to see the exact scaling law of multiuser diversity gain, this section provides lower and upper bounds on the asymptotic average achievable capacity of the secondary network and derives the
scaling law of the average achievable capacity of the secondary network from the two bounds.

\subsection{A Lower Bound}
To give the independency between $\alpha_{s,i}$ and $\beta_{s,i}$ in
the selection of a secondary transmitter, scheduling procedure is
divided into two stages. In the first stage, the secondary receiver
determines a set $\mathcal{S}$ consisting of the candidates for
transmission which have low enough $\alpha_{s,i}$ so that the
calculated $\underset{i}{\max} P_{\mu,i} >P_{s,\max}$, and hence the transmit
power is saturated with $P_{s,\max}$. In the second stage, if the
set $\mathcal{S}$ is not empty, the $i^{\dagger}$ secondary
transmitter which has the largest $\beta_{s,i}$ among the elements
of $\mathcal{S}$ is finally scheduled with its maximum allowable
transmit power $P_{s,\max}$. Otherwise, the $i^{\dagger}$th
secondary transmitter which has the maximum $\gamma_{s,i}$ among the
$N$ secondary transmitters is scheduled with the transmit power of
$P_{\mu,i^{\dagger}}$. Hence, the index of the scheduled secondary
transmitter is given by
\begin{align}
i^{\dagger} = \left\{ \begin{array}{ll}
\arg\underset{i}{\max} \gamma_{s,i} ,& \underset{i}{\max}P_{\mu,i}\le P_{s,\max}\\
\arg\underset{i\in\mathcal{S}}{\max} \beta_{s,i} ,&
\underset{i}{\max}P_{\mu,i}> P_{s,\max}
\end{array}\right..
\end{align}
This two-stage scheduling certainly provides a lower bound on the
average achievable capacity of the secondary network, which is given
in the following theorem. \vspace{0.05in}
\begin{theorem}
For finite $P_{s,\max}$, a lower bound on the average achievable
capacity of the secondary network converges to
\begin{align}\label{Cap_low_complete}
&\mathbb{E}[C_{s,low}] = e^{\frac{K}{P_{s,\max}\ln(1\!-\!\frac{1}{N})}}\Bigg[\frac{\lambda}{\mu_{\min}}(1-p_d)\Bigg(\! \log_2(1\!+\!b_{N,low1})\nonumber\\
&-\frac{e^{\frac{1}{P_p}}\textrm{E}_1\!\!\left(\frac{1}{P_p}\right)}{\ln 2} \Bigg) + \left(\! 1\!-\!\frac{\lambda}{\mu_{\min}} \!\right)(1-p_f)\log_2(1+b_{N,low1}) \Bigg] \nonumber\\
&+ \left(1-e^{\frac{K}{P_{s,\max}\ln(1-\frac{1}{N})}} \right)\Bigg[\frac{\lambda}{M_{avg,l}}(1-p_d)\nonumber\\
&\times\Bigg( \log_2 \left( 1+P_{s,\max}\ln\left( N\left( 1-e^{-\frac{K}{P_{s,\max}}} \right) \right) \right)\nonumber\\
&-\frac{e^{\frac{1}{P_p}}\textrm{E}_1\left(\frac{1}{P_p}\right)}{\ln 2} \Bigg)+\left(1-\frac{\lambda}{M_{avg,l}}\right)(1-p_f)\nonumber\\
&\times\log_2 \left( 1+P_{s,\max}\ln\left( N\left( 1-e^{-\frac{K}{P_{s,\max}}} \right) \right) \right)\Bigg]
\end{align}
as $N$ goes to infinity, where
\begin{align}
b_{N,low1}=P_{s,\max}\mathcal{W}\left(\frac{KN}{P_{s,\max}}e^{\frac{K}{P_{s,\max}}}\right)-K,
\end{align}
\begin{align}
M_{avg,l} = p_de^{-R_p}\!+\!(1\!-\!p_d)\frac{\left(\!
e^{\frac{K(R_pP_{s,\max}+1)}{P_{s,\max}}}-1
\!\right)e^{-R_p(1\!+\!K)}}{\left(\!
e^{\frac{K}{P_{s,\max}}}\!-\!1\! \right)\left( R_pP_{s,\max}\!+\!1
\right)}
\end{align}\label{theorem3}
\end{theorem}where $\mathcal{W}(\cdot)$ denotes a Lambert W function.
\begin{proof}
The proof is given in Appendix.
\end{proof}\vspace{0.05in}

As shown in (\ref{Cap_low_complete}), the lower bound can be divided
into the capacities when the transmit power of the scheduled user is
not bounded and when the transmit power is bounded by $P_{s,\max}$.
The following corollary further simplifies the lower bound.
\vspace{0.05in}
\begin{corollary}\label{C_low_infty}
If we assume that $N$ goes to infinity, the result of Theorem
\ref{theorem3} is approximated as
\begin{align}
&\mathbb{E}[C_{s,low}]\approx \frac{\lambda}{M_{avg,l}}(1-p_d)\log_2(\ln N)\nonumber\\
&~~~\qquad\qquad+ \left( 1-\frac{\lambda}{M_{avg,l}} \right)(1-p_f)\log_2(\ln N)\nonumber\\
&=\left[ \frac{\lambda}{M_{avg,l}}(1-p_d) + \left( 1-\frac{\lambda}{M_{avg,l}} \right)(1-p_f) \right]\log_2(\ln N).
\end{align}
\end{corollary}
\begin{proof}
In (\ref{Cap_low_complete}), the probability of a secondary
transmission with unbounded power,
$e^{\frac{K}{P_{s,\max}\ln(1-\frac{1}{N})}}$, decreases with $N$.
\end{proof}\vspace{0.05in}

Corollary \ref{C_low_infty} indicates that the scaling law of the
lower bound is $k_{l}\log_2(\ln N)$, where $0\le k_{l}\le 1$ is
constant determined by system parameters.


\subsection{An Upper Bound}
An upper bound can be obtained from the case that
$\underset{i}{\max}\beta_{s,i}$ and $\underset{i}{\min}\alpha_{s,j}$ are the effective
forwarding and interference channel gains of the scheduled secondary
transmitter, respectively. This scenario certainly constructs an
upper bound on the achievable capacity of the secondary network and
the effective SNR of the scheduled secondary transmitter in this
case is given by
\begin{align}\label{snr_type}
\gamma_{s,up} = \left\{ \begin{array}{ll}
\frac{K\underset{i}{\max}\beta_{s,i}}{\underset{j}{\min}\alpha_{s,j}} ,& \underset{i}{\max}P_{\mu,i}\le P_{s,\max}\\
P_{s,\max}\underset{i}{\max}\beta_{s,i} ,& \underset{i}{\max}P_{\mu,i}> P_{s,\max}
\end{array}\right..
\end{align}

The following theorem shows an asymptotic behavior of the upper
bound on the average achieved capacity. \vspace{0.05in}
\begin{theorem}
For finite $P_{s,\max}$, the upper bound of the secondary average
achievable capacity converges to
\begin{align}\label{Cap_up_complete}
&\mathbb{E}[C_{s,up}]= e^{\frac{K}{P_{s,\max}\ln(1\!-\!\frac{1}{N})}}\Bigg[ \frac{\lambda}{\mu_{\min}}(1-p_d)\Bigg(\!\log_2 (P_{s,\max}\ln N)\nonumber\\
&-\!\frac{\textrm{E}_1\!\!\left(\frac{1}{P_p}\right)e^{\frac{1}{P_p}}}{\ln 2} \Bigg)+\left( \!1\!-\!\frac{\lambda}{\mu_{\min}} \right)(1-p_f)\log_2(P_{s,\max}\ln N) \Bigg]\nonumber\\
&+\left( 1-e^{\frac{K}{P_{s,\max}\ln(1-\frac{1}{N})}} \right)\Bigg[ \frac{\lambda}{\textrm{M}_{avg,u}(N)}(1-p_d)\nonumber\\
&\times\left( \log_2(1+P_{s,\max}\ln N)-\frac{\textrm{E}_1\left(\frac{1}{P_p}\right)e^{\frac{1}{P_p}}}{\ln 2} \right)\nonumber\\
&+\left(\!1- \frac{\lambda}{\textrm{M}_{avg,u}(N)}
\!\right)(1-p_f)\log_2(1+P_{s,\max}\ln N) \Bigg]
\end{align}
as $N$ grows to infinity, where
\begin{align}
&\textrm{M}_{avg,u}(N)\nonumber\\
&= p_de^{-R_p}+(1-p_d)\frac{ Ne^{-R_p(1+K)}\left( e^{\frac{K(R_pP_{s,\max}+N)}{P_{s,\max}}}-1 \right) }{ (R_pP_{s,\max}+N)\left( e^{\frac{KN}{P_{s,\max}}}-1 \right) }.
\end{align}
\label{theorem4}
\end{theorem}
\begin{proof}
The proof is similar to that of theorem 1.
\end{proof}\vspace{0.05in}

The asymptotic upper bound on the average achievable capacity can be
divided into the two parts according to the type of received SNR
described in (\ref{snr_type}). The following corollary further
simplifies the lower bound. \vspace{0.05in}
\begin{corollary}\label{corollary3}
If we assume that $N$ goes to infinity, the asymptotic upper bound
on the average achievable capacity in (\ref{Cap_up_complete}) is
given by
\begin{align}
&\mathbb{E}[C_{s,up}]\approx \frac{\lambda}{e^{-R_p}}(1-p_d)\log_2(\ln N)\nonumber\\
&~~~\qquad\qquad+ \left( 1-\frac{\lambda}{e^{-R_p}} \right)(1-p_f)\log_2(\ln N)\nonumber\\
&=\left[ \frac{\lambda}{e^{-R_p}}(1-p_d) + \left(1-\frac{\lambda}{e^{-R_p}} \right)(1-p_f) \right]\log_2(\ln N).
\end{align}
\end{corollary}
\begin{proof}
In (\ref{Cap_up_complete}), the capacity via bounded secondary transmit power,
$\gamma_{s,up}=P_{s,\max}\underset{i}{\max}\beta_{s,i}$, becomes
more dominant than the other as $N$ grows.
\end{proof}\vspace{0.05in}

Corollary \ref{corollary3} indicates that the upper bound grows like
$k_u\log_2(\ln N)$, where $0\le k_u\le 1$ is a constant determined
by system parameters. Consequently, we can characterize the
multiuser diversity gain of the secondary network under finite
$P_{s,\max}$ from the scaling law of the lower and upper bounds. The
asymptotic capacity of the secondary network is laid between the
asymptotic lower and upper bounds as
\begin{align}
k_l\log_2(\ln N)\le \mathbb{E}[C_s] \le k_u\log_2(\ln N).
\end{align}

This result implies that the scaling law of secondary network is $k\log_2(\ln N)$, where $k_l\le k \le k_u$.
This scaling law is similar to that in underlaid CR system, $\log_2(\ln N)$, except a constant scaling term $k$ determined by system parameters \cite{MUD_SS}.

\begin{figure}[!t]
    \centerline{\psfig{figure=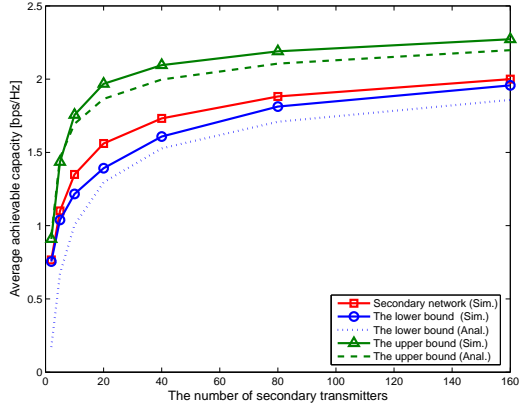,width=0.9\columnwidth} }
    \caption{Average capacities of the secondary network ($p_d=0.8, p_f=0.3, \lambda=0.5, P_p=10\textrm{dB}, P_{s,\max}=P_p, R = 0.5, \mu_{\min}=0.95$)}
    \label{Fig:feasible_power}
\end{figure}

Fig. \ref{Fig:feasible_power} shows the average achievable capacity
of the secondary network, the lower bound, and the upper bound
versus $N$ for $P_{s,\max}=P_p$. It is verified that the lower and
upper bounds on the secondary network capacity are quite well
characterized by the asymptotic approximations even if the number of
secondary transmitter is small. In addition, this figure confirms
that the scaling laws of lower and upper bounds are the same as that
of the exact capacity.

\section{Conclusion}
This paper has investigated the multiuser diversity gain and its
capacity scaling law in an interweaved CR system. We have analyzed the capacity of a secondary network by taking into account
the traffic dynamics of a primary user and reliability of spectrum
sensing in the secondary network. Then, by applying extreme value
theory, we have derived asymptotic capacity of the secondary network
and characterized the multiuser diversity gain. The asymptotic capacity has
been shown to grow like $k\log_2(\ln N)$, where $0\le k\le 1$.
This scaling law of the capacity is similarly observed in underlaid CR systems even
though the mechanism of spectrum access is totally different. Our
numerical results have also verified that the asymptotic capacity
agrees well with the exact capacity even for small number of
secondary transmitters, $N$.

\appendix[Proof of Theorem 1]

The average achievable capacity by the suboptimal two stage
scheduling is given by
\begin{align}\label{cap_sub_sec}
\mathbb{E}[C_{s,low}]=&\Pr\left[ \max_{i}P_{\mu,i}\le P_{s,\max}
\right]\mathbb{E}[C_{s,low}^{(1)}]\nonumber\\
& + \Pr\left[ \max_{i}P_{\mu,i}>
P_{s,\max} \right]\mathbb{E}[C_{s,low}^{(2)}]
\end{align}
where $\mathbb{E}[C_{s,low}^{(1)}]$ and
$\mathbb{E}[C_{s,low}^{(2)}]$ are the average achievable capacities
when there is no element in $\mathcal{S}$ and when there exists at
least one element in $\mathcal{S}$ and given, respectively, by
\begin{align}\label{cap_sub_sec_1st}
\mathbb{E}&[C_{s,low}^{(1)}]=\frac{\lambda}{\mu_{\min}}(1-p_d)\nonumber\\
&~\times\mathbb{E}\left[ \log_2 \left( 1+ \frac{\underset{i}{\max}\frac{K\beta_{s,i}}{\alpha_{s,i}}}{1+\beta_pP_p} \right)\bigg|  \max_i \frac{K}{\alpha_{s,i}}\le P_{s,\max} \right]\nonumber\\
&~+\left( 1-\frac{\lambda}{\mu_{\min}}\right)(1-p_f)\nonumber\\
&~\times\mathbb{E}\left[ \log_2
\left(1+\max_{i}\frac{K\beta_{s,i}}{\alpha_{s,i}} \right)\bigg|
\max_i\frac{K}{\alpha_{s,i}}\le P_{s,\max} \right],
\end{align}
\begin{align}\label{cap_sub_sec_2nd}
\mathbb{E}&[C_{s,low}^{(2)}]=\frac{\lambda}{\mathbb{E}\left[ \mu(\alpha_{s,i^{\dagger}}, P_{s,\max})\Big|\underset{i}{\max}P_{\mu,i}>P_{s,\max} \right]}(1-p_d)\nonumber\\
&~~\times\mathbb{E}\left[ \log_2\left( 1+\frac{\underset{i\in\mathcal{S}}{\max}\beta_{s,i}P_{s,\max} }{ 1+\beta_pP_p } \right) \right]\nonumber\\
&~~+\left(1-\frac{\lambda}{\mathbb{E}\left[
\mu(\alpha_{s,i^{\dagger}},X
P_{s,\max})\Big|\underset{i}{\max}P_{\mu,i}>P_{s,\max}
\right]}\right)\nonumber\\
&~~\times(1-p_f)\mathbb{E}\left[ \log_2 \left( 1+
\underset{i\in\mathcal{S}}{\max}\beta_{s,i}P_{s,\max}\right)
\right].
\end{align}
As $N$ grows, the distribution of maximum random variables in
(\ref{cap_sub_sec_1st}) and (\ref{cap_sub_sec_2nd}) are converges to
the Gumbel distribution or Frech\'et distribution
\cite{Order_statistic}. Using these approximated distributions, we
can asymptotically approximate the expectations as shown in
(\ref{Cap_low_complete}).

\cleardoublepage

\end{document}